\documentclass{AIMS}
\usepackage{amssymb}
\usepackage{amsmath}
\usepackage{paralist}
\usepackage{hyperref}
\newcommand*{\doi}[1]{\href{http://dx.doi.org/#1}{doi: #1}}
\hypersetup{urlcolor=blue, citecolor=red}

\newtheorem{theorem}{Theorem}[section]

\newtheorem{lemma}[theorem]{Lemma}
\newtheorem{proposition}{Proposition}

\theoremstyle{definition}

\newtheorem{remark}{Remark}

\usepackage{graphicx}
\usepackage{enumerate}
\usepackage{xcolor}
\usepackage{tikz}

\usepackage[numbers]{natbib}

\newcommand{\bA}{\mathbf{A}}
\newcommand{\bB}{\mathbf{B}}
\newcommand{\bC}{\mathbf{C}}
\newcommand{\bD}{\mathbf{D}}

\newcommand{\bI}{\mathbf{I}}
\newcommand{\bK}{\mathbf{K}}
\newcommand{\bM}{\mathbf{M}}
\newcommand{\bN}{\mathbf{N}}
\newcommand{\bQ}{\mathbf{Q}}
\newcommand{\bR}{\mathbf{R}}
\newcommand{\bS}{\mathbf{S}}
\newcommand{\bU}{\mathbf{U}}
\newcommand{\bW}{\mathbf{W}}
\newcommand{\bX}{\mathbf{X}}
\newcommand{\ba}{\mathbf{a}}
\newcommand{\be}{\mathbf{e}}
\newcommand{\bbf}{\mathbf{f}}
\newcommand{\bg}{\mathbf{g}}
\newcommand{\bn}{\mathbf{n}}
\newcommand{\bu}{\mathbf{u}}
\newcommand{\bv}{\mathbf{v}}
\newcommand{\bx}{\mathbf{x}}
\newcommand{\bzero}{\mathbf{0}}
\newcommand{\bSigma}{\mathbf{\Sigma}}

\newcommand{\real}{\mathbb{R}}
\newcommand{\complex}{\mathbb{C}}
\newcommand{\Lapl}{\mathcal{L}}

\newcommand{\btK}{\widetilde{\mathbf{K}}}
\newcommand{\btC}{\widetilde{\mathbf{C}}}

\newcommand{\btf}{\widetilde{\mathbf{f}}}
\newcommand{\btu}{\widetilde{\mathbf{u}}}

\newcommand{\diag}{\mathrm{diag}\,}
\newcommand{\norm}[1]{\left\|{#1}\right\|}

\renewcommand\Re{\operatorname{Re}}
\renewcommand\Im{\operatorname{Im}}

\begin{document}

\title[Rayleigh Damped Networks]{Characterization and synthesis of\\
Rayleigh Damped Elastodynamic Networks}
\author[A. Gondolo \and F. Guevara Vasquez]{}
\subjclass{Primary: 74B05, 35R02.}
\keywords{Elastodynamic networks, Response function, Damping, Network synthesis, Proportional damping}

\email{agondolo@math.utah.edu}
\email{fguevara@math.utah.edu}

\centerline{\scshape Alessandro Gondolo and Fernando Guevara Vasquez}
\medskip
{\footnotesize
 \centerline{Mathematics Department, University of Utah}
   \centerline{155 S 1400 E RM 233}
   \centerline{Salt Lake City, UT 84112-0090, USA}
}

\bigskip


\begin{abstract}
We consider damped elastodynamic networks where the damping matrix is assumed to be a non-negative linear combination of the stiffness and mass matrices (also known as Rayleigh or proportional damping). We give here a characterization of the frequency response of such networks. We also answer the synthesis question for such networks, i.e., how to construct a Rayleigh damped elastodynamic network with a given frequency response. Our analysis shows that not all damped elastodynamic networks can be realized when the proportionality constants between the damping matrix and the mass and stiffness matrices are fixed.
\end{abstract}

\maketitle

\section{Introduction} \label{sec:intro}

The second Newton's law applied to a network of springs, masses, and dampers gives
\begin{equation}
 \bK \bu + \bC\dot{\bu} + \bM \ddot{\bu} = \bbf,
 \label{eq:2ndnewt}
\end{equation}
where $\bu$ is the vector of displacements of the network nodes and $\bbf$ is the vector of external forces. The stiffness, damping, and mass matrices are $\bK$, $\bC$, and $\bM$, respectively. These matrices are defined in section~\ref{sec:fr}. We are mainly concerned with networks where the damping is {\em proportional} or of {\em Rayleigh type}, which means that there are some known proportionality constants $\alpha,\beta \geq 0$ linking the damping matrix to the stiffness and mass matrices:
\begin{equation}
 \bC = \alpha \bK  + \beta \bM.
\label{eq:rdamp}
\end{equation}

The Rayleigh damping assumption means that:  
 \begin{enumerate}[(a)]
  \item The springs have a dashpot in parallel with damping constant proportional to the spring constant. One way of achieving this is to construct springs from Kelvin-Voigt solids (see e.g., \cite[\S 7.3]{Chung:2007:GCM}), where the damping constant is proportional to the stiffness constant. 
  \item Each mass lies in a fixed cavity filled with a viscous liquid, where the viscosity constant is proportional to the mass, or the cavity size is adjusted to obtain the same effect (i.e., inversely proportional to the mass). 
 \end{enumerate} 
 The proportionality constants $\alpha$ (linking the damping to the stiffness) and $\beta$ (linking the damping to the mass) are assumed to be the same for all the springs and nodes in the network. An example of such a network is given in figure~\ref{fig:net}.

\begin{figure}
 \begin{center}
  \includegraphics[width=.55\textwidth]{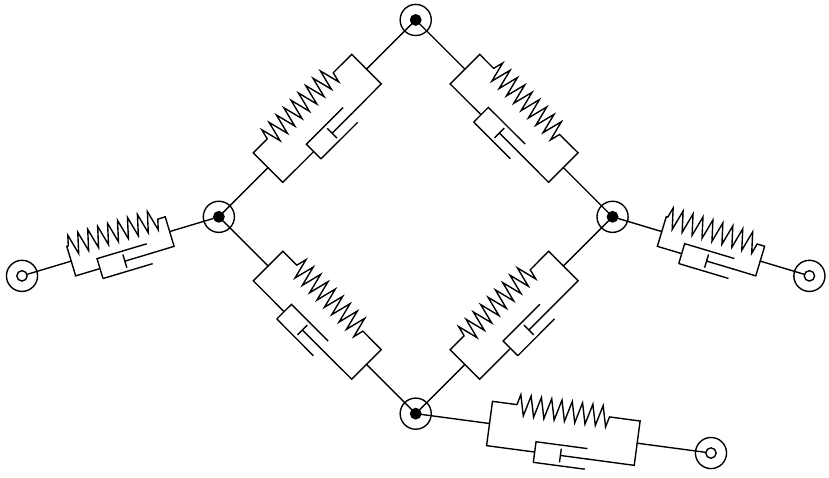}
 \end{center}
 \caption{An example of a Rayleigh damped network. For the terminal nodes the masses are colored in white and for the internal nodes in black. Each mass is surrounded by a cavity containing a viscous liquid. The damping coefficient of each of the linear dampers is $\alpha$ times the stiffness constant of the corresponding spring, and the viscous damping coefficient of each mass is $\beta$ times the corresponding mass.}
 \label{fig:net}
\end{figure}

We answer two questions about the (frequency) response of networks with this class of damping. The response is the frequency-dependent linear relationship between the displacements and forces at a few  terminal, or accessible, nodes. The first question we answer is the characterization question, i.e., we give the form of all possible responses for this particular class of networks. The second question is the synthesis: can we build a network from this class that mimics any admissible response?

For static networks (zero frequency), the characterization and synthesis questions were established by \citet{Camar:2003:DCS}, as part of a characterization of the possible macroscopic behaviors of a static elastic material under a single displacement field. Then, \citet{Milton:2008:RRM} characterized the response of damped elastic networks at a single, non-resonant frequency, and found it is possible to build a damped elastic network that mimics a prescribed response at {\em one single frequency}. The complete characterization and synthesis for the undamped case is done in \cite{Guevara:2011:CCS}: it is shown that the response is a matrix with rational function (of frequency) entries, and that given the response for an undamped elastic network, it is possible to build a network that mimics this response for {\em all non-resonant frequencies}. The new characterization and synthesis results that we present here are a generalization of the results in \cite{Guevara:2011:CCS} to damping of Rayleigh type. In particular, our results show that the Rayleigh damping model is incomplete, in the sense that it cannot describe, by itself, the responses of all possible elastodynamic networks with damping.

The characterization and synthesis for elastic networks have been used by
\citet{Camar:2003:DCS} to show how to design a (linear) elastic material that
has a prescribed response under a single displacement field. A similar
technique exploiting the characterization of networks of resistors was used by
\citet{Camar:2002:CSD} to show how to design a conductor that has a prescribed
response under a single voltage field. Synthesis results are also known for
electrical networks with resistors (Kirchhoff's $Y-\Delta$ theorem,
\citet*{Curtis:1998:CPG} for planar resistor networks); networks of resistors,
capacitors and inductors (Foster \cite{Foster:1924:RT,Foster:1924:TRD}, 
\citet{Bott:1949:ISU}, \citet{Milton:2008:RRM}); acoustic networks (\citet{Milton:2008:RRM}); and an
electromagnetic version of elastodynamic networks (Milton and Seppecher
\cite{Milton:2010:EC,Milton:2010:HEC}).

We start in section~\ref{sec:fr} by defining the stiffness, mass, and damping matrices for a damped elastodynamic network and, given only access to a few terminal nodes, its frequency response. Then in section~\ref{sec:massless}, we show how the characterization and synthesis results for static elastic networks in \cite{Camar:2003:DCS,Guevara:2011:CCS} can be generalized to massless elastic networks with Rayleigh damping. For general Rayleigh damped networks, the characterization result appears in section~\ref{sec:char} and the synthesis result in section~\ref{sec:synth}. We finish by giving in section~\ref{sec:resonances} the loci of the resonances of Rayleigh damped network.

\section{The frequency response of an elastodynamic network with damping}
\label{sec:fr}

The linearized Hooke's Law for a single spring located between nodes $\bx_1$ and $\bx_2$ relates the forces $\ba_i$ (supported at node $\bx_i$) to the displacements $\bu_i$ (assuming these are small) through the relation
\[
 \ba_2 = - k_{1,2} \frac{(\bx_2 - \bx_1)(\bx_2 - \bx_1)^T}{\norm{\bx_2 -
 \bx_1}^2}
 (\bu_2 - \bu_1) = -\ba_1.
\]
Let $\bK$ be the stiffness matrix, $\bC$ be the damping matrix, and $\bM$ be a diagonal matrix with the masses of the nodes in the diagonal. In the case of a single spring, the stiffness and mass matrices are
\[
 \begin{aligned}
 &\bK = k_{1,2} \begin{bmatrix} \bn_{1,2} \bn_{1,2}^T & -\bn_{1,2}
 \bn_{1,2}^T\\ -\bn_{1,2} \bn_{1,2}^T & \bn_{1,2} \bn_{1,2}^T
 \end{bmatrix},
 \quad
 \bM = \diag(m_1\be,m_2\be),\\
 &\bn_{1,2} = \frac{\bx_2 - \bx_1}{\norm{\bx_2 - \bx_1}}, 
 \end{aligned}
\]
where $\be = (1,\ldots,1)^T \in \real^d$, $d=2,3$. For a network with $N$ nodes, the mass matrix is an $Nd \times Nd$ diagonal matrix that is defined similarly with the nodal masses $m_1,\ldots,m_N$. The stiffness and damping matrices of a network are $Nd \times Nd$ matrices obtained by adding the contributions from individual springs or dashpots:
\begin{equation}
 \bK = \sum_{\text{spring~}ij} k_{ij} \bN_{ij} 
 ~\text{and}~
 \bC = \sum_{\text{spring~}ij} c_{ij} \bN_{ij} 
\end{equation}
where $k_{ij}$ (resp.\ $c_{ij}$) is the spring (resp.\ damping) constant for the spring (resp.\ damper) between nodes $i$ and $j$ and 
\begin{equation}
 \bN_{ij} = \left( \begin{bmatrix} \be_i& \be_j \end{bmatrix} 
  \begin{bmatrix} 1 & -1\\ -1 & 1 \end{bmatrix}
  \begin{bmatrix} \be_i^T\\ \be_j^T \end{bmatrix} \right) \otimes ( \bn_{ij} \bn_{ij}^T).
\end{equation}
Here we used the Kronecker product $\otimes$ and the $i-$th canonical basis vector $\be_i \in \real^{N}$. As in the case of a single spring, the vector  $\bn_{ij} \in \real^d$ is a unit length vector with direction $\bx_i -\bx_j$.

For a network of springs, masses, and dashpots, the balance of forces (Newton's
second law) gives a system of ODEs, which is given in \eqref{eq:2ndnewt}.  Now,
recall the Laplace transform:
\[
 \btu(\lambda) = \Lapl[\bu(t)](\lambda)=\int_0^\infty e^{-\lambda t}\bu(t)dt.
\]
Applying the Laplace transform to \eqref{eq:2ndnewt}, the ODE becomes
\begin{equation} \label{eq:allterm}
 (\bK + \lambda\bC + \lambda^2\bM)\btu = \btf.
\end{equation}
In the following, we only work in the Laplace domain; the tilde notation is dropped for simplicity. Our results can also be formulated in the frequency domain by using the Fourier transform instead of the Laplace transform on \eqref{eq:2ndnewt}, or simply by setting $\lambda = \imath \omega$.

For the purpose of introducing the problem, we first consider the case where all nodes have a non-zero mass. Let us partition the nodes of the network into terminal nodes ($B$) and interior nodes ($I$). At the interior nodes, the external forces are zero; the displacement of the interior nodes is determined by the displacement of the boundary nodes, except at the few frequencies corresponding to the resonances of the system. If subscripts $B$ and $I$ denote the quantities referring to their respective nodes, then the displacements at the boundary nodes are related to the forces at the boundary nodes by
\[
 \bbf_B = \bW(\lambda) \bu_B,
\]
where the displacement to forces map is 
\begin{equation} \label{eq:w}
 \begin{aligned}
 \bW(\lambda) &= \bK_{BB} + \lambda\bC_{BB}+ \lambda^2\bM_{BB} \\
 &- (\bK_{BI} + \lambda\bC_{BI}) (\bK_{II} + \lambda\bC_{II} + \lambda^2\bM_{II})^{-1}(\bK_{IB}+\lambda\bC_{IB}),
 \end{aligned}
\end{equation}
which is the Schur complement of the $II$ block in the matrix $\bK + \lambda \bC + \lambda^2 \bM$ (see Appendix~\ref{app:schur} for the definition of the Schur complement and properties).

\begin{remark}
We assume throughout this paper that each spring-damper pair occupies an arbitrarily small volume surrounding the segment connecting the corresponding nodes. We also assume that the cavities with viscous fluid (and hence the masses) can be made arbitrarily small.  We need these assumptions to constrain the internal nodes in our synthesis result to be within an $\epsilon-$neighborhood of the convex hull of the terminal nodes.
\end{remark}

\section{Massless Rayleigh damped networks}
\label{sec:massless}

First, we characterize the networks with Rayleigh damping and no mass. Later, in section~\ref{sec:char}, we use such networks to simplify the work for general networks. The theorems in this section are a natural extension of the characterization and synthesis in \cite{Camar:2003:DCS,Guevara:2011:CCS} for static networks, i.e., networks with springs only. In the static case the response matrix is given by \eqref{eq:w} with $\lambda = 0$.  The forces $\bbf = ( \bbf^T_1, \ldots, \bbf^T_n )^T$ at the nodes $\bx_1,\ldots,\bx_n$ form a {\em balanced system of forces} when
\begin{equation}
 \sum_{i=1}^n \bbf_i = \bzero 
 ~ \text{(equilibrium of forces) and}~
 \sum_{i=1}^n \bx_i \times \bbf_i = \bzero ~  \text{(equilibrium of torques)},
 \label{eq:balanced}
\end{equation}
where $\bu \times \bv = (u_2 v_3 - u_3 v_2, u_3 v_1 - u_1 v_3, u_1 v_2 - u_2 v_1)$ if $d=3$ and $\bu \times \bv = u_1 v_2 - u_2 v_1$ if $d=2$, is the usual cross product. The characterization and synthesis theorems in \cite{Camar:2003:DCS,Milton:2008:RRM,Guevara:2011:CCS} are summarized in the next theorem.

\begin{theorem} \label{spring_static}
The response matrix $\bW$ of any network of springs is symmetric positive semidefinite where each column is a balanced system of forces at the terminal nodes. (Characterization)

Conversely, any matrix $\bW$ that is symmetric positive semidefinite where each column is a balanced system of forces at the terminals can be realized by a network of only springs. Moreover the internal nodes of such network can be chosen so as to avoid a finite number of points and to be within an $\epsilon-$neighborhood of the convex hull of the terminal nodes. (Synthesis)
\end{theorem}
\begin{proof}
See e.g. \cite[Lemma 2]{Guevara:2011:CCS} for the characterization and \cite[Theorem 1]{Guevara:2011:CCS} for the synthesis.
\end{proof}

Theorem~\ref{spring_static} can be readily extended to massless Rayleigh damped networks.

\begin{theorem} \label{thm:massless_Rayleigh}
The response of a massless network with Rayleigh damping is of the form
\begin{equation} \label{eq:massless_W}
 (1+\alpha\lambda)\bW,
\end{equation}
where $\alpha\geq0$, $\lambda$ is the Laplace parameter, and $\bW$ is the response of a static network, i.e., $\bW$ is symmetric positive semidefinite with each column being a balanced system of forces at the terminals $\bx_1,\ldots,\bx_n$ (as in theorem~\ref{spring_static}).

Conversely, given any matrix-valued function of $\lambda$ of the form \eqref{eq:massless_W} where $\bW$ is the response of a static network at the nodes $\bx_1,\ldots,\bx_n$, there is a massless Rayleigh damped network with terminals at the nodes $\bx_1,\ldots,\bx_n$ realizing it. Moreover the internal nodes of such network can be chosen so as to avoid a finite number of points and to be within an $\epsilon-$neighborhood of the convex hull of the terminal nodes.
\end{theorem}
\begin{proof}
The forward direction is due to the homogeneity of degree 1 of the Schur complement (see Appendix~\ref{app:schur}). Let $\bK$ be the stiffness matrix for a static network and $\bW$ the response matrix of the network at some terminal nodes $B$. Then the frequency response of the massless network with Rayleigh damping when all the nodes are terminals is $(1+\alpha\lambda)\bK$. Then by  taking the Schur complement, the response at the nodes $B$ is $(1+\alpha\lambda)\bW$. 

To prove the converse, assume we are given a frequency response of the form $(1+\alpha\lambda)\bW$, where $\bW$ is the response of a static network. By the second part of theorem~\ref{spring_static}, there is a  network of springs with stiffness matrix $\bK$ with response $\bW$ at the nodes $B$. The internal nodes of this network can be chosen so as to avoid a finite number of points and to be in an arbitrarily small neighborhood of the convex hull of the terminals. Then a network that has $(1+\alpha \lambda) \bW$ as its response is the network with response at all the nodes $(1+\alpha \lambda)\bK$, i.e., the network obtained from the second part of theorem~\ref{spring_static} with dashpots in parallel with each spring, each dashpot having a damping constant being $\alpha$ times the stiffness of the associated spring. Here we have used the assumption that each spring and damper pair occupies an arbitrarily thin segment between the corresponding nodes, so the network transformations that do not change the response in \cite[\S 2.3]{Guevara:2011:CCS} are still valid.
\end{proof}

\begin{remark}
Theorem~\ref{thm:massless_Rayleigh} is valid for planar networks, i.e., networks for which all the springs lie in a plane. This is because theorem~\ref{spring_static} is valid for planar networks: one can always realize the response of a planar network with a planar network.
\end{remark}

\section{Characterization of general Rayleigh damped networks}
\label{sec:char}

In this section, we give conditions that the response of Rayleigh damped networks needs to satisfy. In general, we may have massless nodes in the network; these are dealt with in lemma~\ref{l:remove_massless}, where we use the characterization for the massless case from section~\ref{sec:massless} to eliminate any massless interior nodes. Theorem~\ref{thm:rayleigh_char} is the characterization result for general Rayleigh networks. Then Proposition~\ref{prop:sign} shows that the conditions of theorem~\ref{thm:rayleigh_char} are consistent with the characterization at a single frequency found by \citet{Milton:2008:RRM} (a  condition which means, in physical terms, that damping can only consume energy).

Let us partition the interior nodes $I$ into the set of interior nodes $J$ with positive mass and the set of massless nodes $L$. Clearly, $\bM_{JJ}$ is positive definite while $\bM_{LL}=\bzero$. The following lemma is similar to \cite[Lemma 3]{Guevara:2011:CCS} and reduces the characterization problem for networks with massless nodes to networks where each node has mass.

\begin{lemma} \label{l:remove_massless}
Let $\bK$ be a $Nd \times Nd$ stiffness matrix and let $\bM$ be a $Nd \times Nd$ (diagonal) mass matrix, where $N = |B \cup I|$. The response at the terminals is
\[
 \begin{aligned}
 \bW(\lambda)&=\btK_{BB}+\lambda\btC_{BB}+\lambda^2\bM_{BB}\\&-(\btK_{BJ}+\lambda\btC_{BJ})(\btK_{JJ}+\lambda\btC_{JJ}+\lambda^2\bM_{JJ})^{-1}(\btK_{JB}+\lambda\btC_{JB}),
 \end{aligned}
\]
for all frequencies $\lambda$ for which $\det(\btK_{JJ}+\lambda\btC_{JJ}+\lambda^2\bM_{JJ} ) \neq 0$. Here the tilde matrices are submatrices of the matrices
\begin{equation} \label{eq:tildeKC}
 \begin{aligned}
 \btK & = \begin{bmatrix} \btK_{BB} & \btK_{BJ} \\ \btK_{JB} & \btK_{JJ} \end{bmatrix}  = \begin{bmatrix} \bK_{BB} & \bK_{BJ} \\ \bK_{JB} & \bK_{JJ} \end{bmatrix} - \begin{bmatrix} \bK_{BL} \\ \bK_{JL} \end{bmatrix}\bK_{LL}^\dagger\begin{bmatrix} \bK_{LB} & \bK_{LJ}\end{bmatrix} ~\text{and}\\
 \btC &= \begin{bmatrix} \btC_{BB} & \btC_{BJ} \\ \btC_{JB} & \btC_{JJ} \end{bmatrix}  = \begin{bmatrix} \bC_{BB} & \bC_{BJ} \\ \bC_{JB} & \bC_{JJ} \end{bmatrix} - \begin{bmatrix} \bC_{BL} \\ \bC_{JL} \end{bmatrix}\bC_{LL}^\dagger\begin{bmatrix} \bC_{LB} & \bC_{LJ}\end{bmatrix},
 \end{aligned}
\end{equation}
where the symbol $\dagger$ denotes the Moore-Penrose pseudoinverse. Also, $\btK_{JJ}$ and $\btC_{JJ}$ are positive semidefinite.
\end{lemma}
\begin{proof}
The response of the network with terminals $B \cup J$ and internal nodes $L$ is $\btK + \lambda \btC + \lambda^2 \diag(\bM_{BB}, \bM_{JJ})$. Equilibrating forces at the nodes $J$ (i.e., taking the Schur complement of the $JJ$ block) gives the desired result. Note that $\btK_{JJ}$ and $\btC_{JJ}$ are positive semidefinite as they are the Schur complement of positive semidefinite matrices (see e.g., \eqref{eq:schursign}).
\end{proof}

\begin{remark} \label{rem:massless}
For Rayleigh networks, the elimination of massless nodes made in lemma~\ref{l:remove_massless} preserves the Rayleigh damping structure. If we have Rayleigh damping then $\bC = \alpha \bK + \beta \bM$ (at all nodes) and
\[
 \btC = \alpha \btK + \beta \diag(\bM_{BB},\bM_{JJ}),
\]
where the tilde matrices are defined as in lemma~\ref{l:remove_massless}. Indeed a simple calculation gives
\[
 \begin{aligned}
  \btC &= 
  \alpha \begin{bmatrix} \bK_{BB} & \bK_{BJ} \\ \bK_{JB} & \bK_{JJ} \end{bmatrix}
- \alpha^2\begin{bmatrix} \bK_{BL} \\ \bK_{JL} \end{bmatrix}(\alpha\bK_{LL})^\dagger\begin{bmatrix}\bK_{LB} & \bK_{LJ}\end{bmatrix}\\
  &+\beta  \begin{bmatrix} \bM_{BB} & \\ & \bM_{JJ} \end{bmatrix}.\\
 \end{aligned}
\]
\end{remark}

We can now formulate a characterization of the response at the terminal nodes of a Rayleigh damped network.

\begin{theorem} \label{thm:rayleigh_char}
Consider a damped mass-spring network with terminals $\bx_1,\ldots,\bx_n$ with Rayleigh damping, i.e., the damping matrix is of the form $\bC = \alpha \bK + \beta \bM$, where $\alpha, \beta \geq 0$ are given. The displacement-to-forces map of any such network is of the form:
\begin{equation} \label{eq:W_poles}
 \bW(\lambda)=(1+\alpha\lambda)\bA+(\beta\lambda+\lambda^2)\bM-\sum_{j=1}^p \frac{(1+\alpha\lambda)^2\bR^{(j)}}{\sigma_j+\lambda(\alpha\sigma_j+\beta)+\lambda^2},
\end{equation}
where
\begin{enumerate}[i.]
 \item $\bR^{(j)}$ is real symmetric positive semidefinite and $\sigma_j>0$.
 \item $\bM$ is real diagonal positive semidefinite.
 \item $\bA$ is real symmetric positive semidefinite.
 \item There are at most $2p$ distinct poles: namely the roots of the polynomials 
  \[
   q_j(\lambda) = \sigma_j+\lambda(\alpha\sigma_j+\beta)+\lambda^2, ~\text{for}~ j=1,\ldots,p.
  \]
  Moreover, for all roots $\lambda^*$ of $q_j(\lambda)$, we have $\Re(\lambda^*)\leq0$. This is the second law of thermodynamics: damping consumes energy.
 \item The response for $\lambda=0$, i.e.,
  \[
   \bW(0)=\bA-\sum_{j=1}^p \frac{\bR^{(j)}}{\sigma_j},
  \]
  is the response of a static network. The characterization for static elastic networks \cite{Camar:2003:DCS} states that $\bW(0)$ must be real symmetric positive semidefinite with every column being a balanced system of forces (see \eqref{eq:balanced}) at the terminals $\bx_1,\cdots,\bx_n$. 
\end{enumerate}
\end{theorem}
\begin{proof}
We begin by using lemma~\ref{l:remove_massless} to remove the massless interior nodes. Let $I=J\cup L$ where $J$ are the positive mass nodes and $L$ are the massless nodes. Let $\btK$ and $\btC$ be as in \eqref{eq:tildeKC}. Then lemma~\ref{l:remove_massless} gives
\begin{equation} \label{eq:res1}
 \begin{aligned}
  \bW(\lambda) &= \btK_{BB} + \lambda\btC_{BB}+\lambda^2\bM_{BB}\\ &- (\btK_{BJ}+\lambda\btC_{BJ})(\btK_{JJ}+\lambda\btC_{JJ}+\lambda^2\bM_{JJ})^{-1}(\btK_{JB}+\lambda\btC_{JB}).
 \end{aligned}
\end{equation}

Since $\bM_{JJ}$ is real diagonal positive definite, it has a square root $\bM_{JJ}^{1/2}$. Now $\btK_{JJ}$ is real symmetric positive semidefinite and so is $\bM_{JJ}^{-1/2}\btK_{JJ}\bM_{JJ}^{-1/2}$. Then by the spectral theorem, there exists $\bU$ unitary and $\bSigma$ real diagonal with nonnegative entries so that
\[
 \bM_{JJ}^{-1/2}\btK_{JJ}\bM_{JJ}^{-1/2}=\bU\bSigma\bU^T.
\]
Consider the matrix $\bX=\bM_{JJ}^{-1/2} \bU$, which clearly is invertible and is so that
\[
 \bX^T\btK_{JJ}\bX=\bSigma ~\text{ and }~ \bX^T\bM_{JJ}\bX=\bI,
\]
where $\bI$ is the identity matrix (in this case $|J|\times|J|$). The resolvent part of \eqref{eq:res1} becomes 
\begin{equation*}
 \begin{aligned}
  \bQ(\lambda) &= \btK_{JJ}+\lambda\btC_{JJ}+\lambda^2\bM_{JJ}\\
  &= \bX^{-T}(\bX^T\btK_{JJ}\bX + \lambda\bX^T\btC_{JJ}\bX + \lambda^2\bX^T\bM_{JJ}\bX)\bX^{-1}\\
  &= \bX^{-T}(\bSigma + \lambda(\alpha\bSigma+\beta\bI) + \lambda^2\bI)\bX^{-1} \\
  &= \bX^{-T}\bD(\lambda)\bX^{-1}.
 \end{aligned}
\end{equation*}
Here we used the result from remark~\ref{rem:massless}, i.e., that $\btC_{JJ} = \alpha \btK_{JJ} + \beta \bM_{JJ}$.  The matrix $\bD(\lambda)$ is diagonal and therefore easily inverted. Then the resolvent $\left[ \bQ(\lambda) \right]^{-1}$ can be written as
\[
 \begin{aligned}
  \left[ \bQ(\lambda) \right]^{-1}&=\bX\bD(\lambda)^{-1}\bX^{T} \\
  &=\bX\operatorname{diag}\left(\frac{1}{\sigma_i + \lambda(\alpha\sigma_i+\beta) + \lambda^2}\right)\bX^{T},
 \end{aligned}
\]
where $\sigma_i$ are the diagonal elements of $\bSigma$. Let $\bK_{BJ}\bX=[\bv_1\cdots\bv_{|J|}]$. By remark~\ref{rem:massless} we have $\btC_{BB} = \alpha \btK_{BB} + \beta \bM_{BB}$ and $\btC_{BJ} = \alpha \btK_{BJ}$. Therefore the response \eqref{eq:res1} can be written in the form
\[
 \bW(\lambda)=(1+\alpha\lambda)\btK_{BB}+ (\beta\lambda+\lambda^2)\bM_{BB}-\sum_{i=1}^{|J|}\frac{\bv_i \bv_i^T(1+\alpha\lambda)^2}{\sigma_i + \lambda(\alpha\sigma_i+\beta) + \lambda^2}.
\]

Notice that the $\sigma_i$ may be zero. This is the case when the network has a so called ``floppy mode'' (non-zero displacements with zero force required). The same argument in \cite[Lemma 12]{Guevara:2011:CCS} can be used to deal with floppy modes. Indeed Lemma 1 in \cite{Guevara:2011:CCS} can be used to show that $\sigma_i = 0$ implies $\bv_i =0$.
By adding up the non-zero residues $\bv_i \bv_i^T$ that share the same denominator, we get the symmetric positive semidefinite residues $\bR^{(j)}$ of the statement of the theorem. The result follows by noticing that $\bA = \btK_{BB}$ is symmetric positive semidefinite, $\bM_{BB}$ is diagonal with nonnegative entries and at $\lambda=0$, the response
\[
 \bW(0)=\bA-\sum_{j=1}^p\frac{\bR^{(j)}}{\sigma_j},
\]
is the response of a static elastic network.

Finally for the roots $\lambda_{\pm}^{(j)}$ of $\sigma_j + (\alpha \sigma_j + \beta) \lambda + \lambda^2$ we have \[ 2 \Re \lambda_+^{(j)} = - (\alpha \sigma_j + \beta) \leq 0.\] Hence the resonances lie in the left half plane.
\end{proof}

\begin{remark}
 Setting $\lambda=0$ in Theorem~\ref{thm:rayleigh_char} gives the static or zero frequency result in \cite{Camar:2003:DCS}. Setting $\alpha=\beta=0$ (i.e., no damping) gives the result in \cite{Guevara:2011:CCS} for elastodynamic networks with no damping. 
\end{remark}

For the Laplace parameter $\lambda$ in the imaginary axis (i.e., real frequencies), the eigenvalues of the imaginary part of the response should have a sign consistent with the energy losses due to damping (again a manifestation of the second law of thermodynamics). These inequalities are essential to the single-frequency characterization and synthesis of \citet{Milton:2008:RRM}, and hence these inequalities should also hold in our case. The next proposition shows that these inequalities hold automatically for matrix functions of $\lambda$ satisfying the hypothesis of theorem~\ref{thm:rayleigh_char}.
\begin{proposition}\label{prop:sign}
Any matrix function of $\lambda$ of the form \eqref{eq:W_poles} and with the properties of theorem~\ref{thm:rayleigh_char} is such that
\[
 \omega \Im \bW(\imath\omega) \geq 0 ~~ \text{for any $\omega \in \real$},
\]
where the inequality $\bA \geq 0$ for a symmetric matrix $\bA$, means $\bA$ is positive semidefinite.
\end{proposition} 
\begin{proof}
Let us rewrite the matrix-valued function $\bW(\lambda)$ in \eqref{eq:W_poles} as
\[
 \bW(\lambda) = \bW_1(\lambda) + \bW_2(\lambda) + \bW_3(\lambda),
\]
where
\[
 \begin{aligned}
 \bW_1(\lambda) &= (1+\alpha\lambda) \left( \bA - \sum_{j=1}^p \frac{\bR^{(j)}}{\sigma_j} \right)
                 = (1+\alpha\lambda) \bW(0),\\
 \bW_2(\lambda) &= (\beta \lambda + \lambda^2)\bM, \text{~and~}\\
 \bW_3(\lambda) &= \sum_{j=1}^p \left[ (1+\alpha \lambda) - \frac{(1+\alpha \lambda)^2 \sigma_j}{\sigma_j + \lambda(\alpha \sigma_j + \beta) + \lambda^2 } \right] \frac{\bR^{(j)}}{\sigma_j}.
 \end{aligned}
\]
To prove the final result, it is enough to show that $\omega \Im \bW_k(\imath\omega) \geq 0$, $k=1,2,3$. The first two cases are clear since: $\omega \Im \bW_1(\imath \omega) = \alpha \omega^2 \bW(0) \geq 0$ because $\bW(0) \geq 0$; and $\omega \Im \bW_2(\imath \omega) = \beta \omega^2 \bM \geq 0$ because $\bM$ is diagonal with nonnegative entries.

Since $\bR^{(j)}/\sigma_j \geq 0$, we have proved the inequality for the third case if for all $\sigma>0$ the function
\[
 f(\lambda) = (1+\alpha \lambda) - \frac{(1+\alpha \lambda)^2 \sigma}{\sigma + \lambda(\alpha \sigma + \beta) + \lambda^2 }
\]
is such that 
\[
 \omega \Im f(\imath \omega) \geq 0, ~~ \text{for $\omega \in \real$.}
\]
To show this inequality, consider the $2\times 2$ complex symmetric matrix
\[
 \bB(\lambda) = 
 \begin{bmatrix}
  1+\alpha \lambda & 1+\alpha \lambda\\
  1+\alpha \lambda & \frac{\sigma + \lambda(\alpha \sigma + \beta) + \lambda^2}{\sigma}
 \end{bmatrix}.
\]
We start with $\omega \geq 0$. Clearly $\Im \bB(\imath\omega) \geq 0$ since
$\det(\Im \bB(\imath\omega)) = \omega \alpha  \beta/\sigma \geq 0$ and all the
entries of $\Im \bB(\imath\omega)$ are non-negative. Then noticing that
$f(\lambda)$ is the Schur complement of the $2,2$ entry in $\bB(\lambda)$ and
using the property of the Schur complement \eqref{eq:schursign}, we have that
$\Im \bB(\imath\omega) \geq 0 \Rightarrow \Im f(\imath \omega) \geq 0$. So we get the result $\omega f(\imath \omega) \geq 0$ when $\omega \geq 0$.

For $\omega \leq 0$, the same reasoning holds. We only need to check that $\Im \bB(\imath\omega) \leq 0$. This is true because all the entries of $\Im \bB(\imath\omega)$ are non-positive and $\det (\Im \bB(\imath\omega)) = \omega \alpha  \beta/\sigma \leq 0$. Therefore $\omega f(\imath \omega) \geq 0$ when $\omega \leq 0$, which completes the proof.
\end{proof}
\section{Synthesis of general Rayleigh damped networks}
\label{sec:synth}

The basic building block for the synthesis is a Rayleigh damped network with: a rank one response; a complex conjugate pair or a real pair of prescribed resonances; and zero static response. The existence of such a network is stated in lemma~\ref{lem:rankone} and proved at the end of this section (it is a similar construction to that in \cite[Lemma 12]{Guevara:2011:CCS}). The synthesis of a general Rayleigh damped network is done in theorem~\ref{thm:synth}.

\begin{lemma} \label{lem:rankone}
Let $\bx_1,\ldots,\bx_n$ be some terminal nodes and let $\bbf_1, \ldots, \bbf_n$ be an arbitrary system of forces at the corresponding terminals (the forces do not need to be balanced). Then for any $\alpha,\beta \geq 0$ and $\sigma>0$, it is possible to build a Rayleigh damped network with zero static response and with
\begin{equation} \label{eq:rankone}
 \bW(\lambda) = \left[ (1+\alpha\lambda) - \frac{(1+\alpha\lambda)^2\sigma}{\sigma  + (\alpha \sigma + \beta) \lambda + \lambda^2} \right] \bbf \bbf^T,
\end{equation}
where $\bbf = [ \bbf_1^T, \ldots, \bbf_n^T]^T$. The internal nodes of such a network can be chosen to avoid a finite number of points and to be within an $\epsilon-$neighborhood of the convex hull of the terminal nodes.
\end{lemma}

\begin{theorem} \label{thm:synth}
Given any matrix-valued function $\bW(\lambda)$ of the form \eqref{eq:W_poles} and satisfying the conditions of theorem~\ref{thm:rayleigh_char} with $\alpha$ and $\beta$ fixed, there exists a Rayleigh damped network with proportionality constants $\alpha$ and $\beta$ and response $\bW(\lambda)$. The internal nodes of such a network can be chosen to avoid a finite number of points and to be within an $\epsilon-$neighborhood of the convex hull of the terminal nodes.
\end{theorem}

\begin{proof}
The proof relies on the superposition of networks, i.e., using the fact that the response of two networks that share terminal nodes and only terminal nodes, is the sum of the responses of each network (see e.g., \cite[\S 2.4]{Guevara:2011:CCS}). To specify the building blocks needed to realize the matrix-valued function $\bW(\lambda)$ in \eqref{eq:W_poles}, we first rewrite it as:
\[
 \bW(\lambda) = \bW_1(\lambda) + \bW_2(\lambda) + \bW_3(\lambda),
\]
where
\[
 \begin{aligned}
 \bW_1(\lambda) &= (1+\alpha\lambda) \left( \bA - \sum_{j=1}^p \frac{\bR^{(j)}}{\sigma_j} \right)
                 = (1+\alpha\lambda) \bW(0),\\
 \bW_2(\lambda) &= (\beta \lambda + \lambda^2)\bM, \text{~and~}\\
 \bW_3(\lambda) &= \sum_{j=1}^p \left[ (1+\alpha \lambda) - \frac{(1+\alpha \lambda)^2 \sigma_j}{\sigma_j + \lambda(\alpha \sigma_j + \beta) + \lambda^2 } \right] \frac{\bR^{(j)}}{\sigma_j}.
 \end{aligned}
\]
We now show how to realize each term $\bW_1$, $\bW_2$ and $\bW_3$ separately by a network. The superposition principle can then be used to find a network realizing $\bW$. 

We first use theorem~\ref{thm:massless_Rayleigh} to see that a network realizing $\bW_1(\lambda)$ is the network of springs realizing the static response $\bW(0)$ (existence guaranteed by theorem~\ref{spring_static} or \cite{Camar:2003:DCS}), with a damper with damping constant $\alpha$ times the spring constant added in parallel to each spring. To realize $\bW_2$ it suffices to endow each terminal node by the mass dictated by $\bM$ and surrounding it by a cavity such that the resulting damping is $\beta$ times the mass.

We now show that each term in the sum in the expression of $\bW_3$ is realizable. Then the realizability of $\bW_3$ follows from the superposition principle. Let us drop the index $j$ for the sake of simplicity and show that there is a Rayleigh damped network with response:
\[
 \left[ (1+\alpha \lambda) - \frac{(1+\alpha \lambda)^2 \sigma}{\sigma + \lambda(\alpha \sigma + \beta) + \lambda^2 } \right] \frac{\bR}{\sigma},
\]
when $\bR$ is real symmetric positive semidefinite. By the spectral theorem it is possible to find real vectors $\bv_k$ such that
\[
\frac{\bR}{\sigma} = \sum_{k=1}^r \bv_k \bv_k^T,
\]
where $r>0$ is the rank of $\bR$. (The case $r=0$ is the trivial $\bR=\bzero$ case).
Hence we have reduced the problem to that of finding a network with zero static response, a rank one response and as resonances the roots of $p(\lambda) = \sigma + \lambda(\alpha \sigma + \beta) + \lambda^2$. Lemma~\ref{lem:rankone} shows how to build such a network. This completes the construction.
\end{proof}

We now prove Lemma~\ref{lem:rankone}.

\begin{proof}
The main idea here is to find a massless Rayleigh damped network and add appropriate masses to the internal nodes to get the desired resonances. According to \cite[Lemma 12]{Guevara:2011:CCS}, it is possible to choose two distinct nodes $\bx_{n+1}$ and $\bx_{n+2}$ and two forces $\bbf_{n+1}$ and $\bbf_{n+2}$ such that the system of forces $\bbf_1,\ldots,\bbf_{n+2}$ is balanced when supported at the nodes $\bx_1,\ldots,\bx_{n+2}$, regardless of the choice of $\bbf$. Then by theorem~\ref{thm:massless_Rayleigh}, there exists a massless Rayleigh damped network with rank one response $(1+\alpha\lambda)[\bbf^T,\bg^T]^T [\bbf^T,\bg^T]$, where $\bg = [\bbf_{n+1}^T,\bbf_{n+2}^T]^T$. We now endow the two internal nodes $\bx_{n+1}$ and $\bx_{n+2}$ with the same mass $m$, that is determined later to match the desired resonances. We also surround the nodes $\bx_{n+1}$ and $\bx_{n+2}$ by a cavity with a viscous fluid, designed so that the damping term is $\beta m$. Then Newton's second law becomes
\[
 \left(\begin{bmatrix}\bbf \\ \bg\end{bmatrix} \begin{bmatrix}\bbf^T & \bg^T\end{bmatrix} (1+\alpha\lambda) + \begin{bmatrix} \bzero & \\ & m\bI\end{bmatrix}(\beta\lambda+\lambda^2)\right)\begin{bmatrix}\bu_B \\ \bu_I\end{bmatrix} = \begin{bmatrix}\bW(\lambda)\bu_B \\ \bzero\end{bmatrix},
\]
where $\bu_B$ and $\bu_I$ are the displacements of the boundary nodes $\bx_1,\ldots,\bx_n$ and interior nodes $\bx_{n+1}$ and $\bx_{n+2}$, respectively, and $\bW(\lambda)$ is the frequency response of this network.

Next, we take the Schur complement of the $II$ block to find the response $\bW(\lambda)$:
\[
 \begin{aligned}
  \bW(\lambda) &= (1+\alpha\lambda)\bbf\bbf^T+(1+\alpha\lambda)\left(\frac{-(1+\alpha\lambda)|\bg|^2}{(1+\alpha\lambda)|\bg|^2+(\beta\lambda+\lambda^2)m}\right)\bbf\bbf^T \\
  &= (1+\alpha\lambda)\bbf\bbf^T-\left(\frac{(1+\alpha\lambda)^2(|\bg|^2/m)}{\lambda^2+(\alpha(|\bg|^2/m)+\beta)\lambda+(|\bg|^2/m)}\right)\bbf\bbf^T.
 \end{aligned}
\]
Let $m = |\bg|^2/\sigma$. Then the response is
\[
 \bW(\lambda) = (1+\alpha\lambda)\bbf\bbf^T-\left(\frac{(1+\alpha\lambda)^2\sigma}{\lambda^2+(\alpha\sigma+\beta)\lambda+\sigma}\right)\bbf\bbf^T,
\]
which is the desired result since $\bW(0) = \bzero$.
\end{proof}

\section{Resonances of Rayleigh damped networks}
\label{sec:resonances}
A natural question to ask is whether it is possible to find a Rayleigh network
with $\alpha$ and $\beta$ fixed with a resonance that is located anywhere in
the left half complex plane. Since the conditions of
theorem~\ref{thm:rayleigh_char} are necessary and sufficient for a matrix
valued function of frequency to be the response of a Rayleigh damped network,
this question is equivalent to finding the set of all possible resonances for
a Rayleigh damped network with $\alpha$ and $\beta$ fixed. Our analysis shows that
this is set is not the left hand plane. This means that the Rayleigh damping model is incomplete, since a general damped network can have resonances anywhere in the left half complex plane.

We state here the loci of the resonances that can be realized with Rayleigh
damped networks. The derivation of these loci is standard and is not included
here.  According to theorem~\ref{thm:rayleigh_char}, the resonances that can be
realized with Rayleigh damped networks are the roots of the quadratic
$\lambda^2 + (\alpha \sigma + \beta) \lambda + \sigma$, i.e.,
\[
 \lambda_{\pm} = \frac{ -(\alpha \sigma + \beta) \pm \sqrt{ (\alpha \sigma + \beta)^2 - 4 \sigma }}{2}.
\]
Because we only can choose $\sigma$, the resonances lie in curves in the complex plane. The specific curves depend on $\alpha$ and $\beta$ and are as follows.

 {\bf Damping at the nodes only ($\alpha=0$):} Union of the segment $\Im \lambda= 0,$ $-\beta \leq \Re \lambda < 0$ and the line $\Re \lambda = - \beta/2$.
 \begin{center}
 \begin{tikzpicture}[scale=0.7]
  \draw[->,thick] (-5,0) to (1,0);
  \draw[->,thick] (0,-2) to (0,2);
  \node[label=below:{Re$\,\lambda$}] at (1,0) {};
  \node[label=right:{Im$\,\lambda$}] at (0,2) {};
  \node at (-2,0) {\rotatebox{90}{$-$}};
  \node at (-4,0) {\rotatebox{90}{$-$}};
  \draw [red,very thick] (-4,0) to (0,0);
  \draw [red,very thick] (-2,-2) to (-2,2);
  \node [label=below:{\textcolor{red}{$-\beta$}}] at (-4,0) {};
  \node [label=above right:{\textcolor{red}{$-\beta/2$}}] at (-2,0) {};
  \draw [red,thick,fill=white] (0,0) circle [radius=0.1];
 \end{tikzpicture}
 \end{center}

 {\bf Damping with dashpots between the nodes only ($\beta=0$):} Union of the negative real axis $\Im \lambda = 0, \Re \lambda <0$ and the circle $|\lambda+\alpha^{-1}| =\alpha^{-1}$ (the origin excepted).
 \begin{center}
 \begin{tikzpicture}[scale=0.7]
  \draw[->,thick] (-5,0) to (1,0);
  \draw[->,thick] (0,-2) to (0,2);
  \node[label=below:{Re$\,\lambda$}] at (1,0) {};
  \node[label=right:{Im$\,\lambda$}] at (0,2) {};
  \node at (-2,0) {\rotatebox{90}{$-$}};
  \node at (-4,0) {\rotatebox{90}{$-$}};
  \draw [red,very thick,domain=0:360] plot ({-2+2*cos(\x)},{2*sin(\x)});
  \draw [red,very thick] (-5,0) to (0,0);
  \draw [black!50, thick,->] (-2,0) -- (-2-1.4142,1.4142) node [midway,above right]{\textcolor{red}{$\alpha^{-1}$}};
  \node [label=below:{\textcolor{red}{$-1/\alpha$}}] at (-2,0) {};
  \draw [red,thick,fill=white] (0,0) circle [radius=0.1];
 \end{tikzpicture}
 \end{center}

 {\bf Case $\alpha, \beta \neq 0$ and $\alpha \beta < 1$:} 
 Union of the negative real axis $\Im \lambda = 0,\Re \lambda <0$ and the circle $|\lambda + \alpha^{-1}| = \alpha^{-1}\sqrt{1-\alpha\beta}$.
 \begin{center}
 \begin{tikzpicture}[scale=0.7]
  \draw[->,thick] (-5,0) to (1,0);
  \draw[->,thick] (0,-2) to (0,2);
  \node[label=below:{Re$\,\lambda$}] at (1,0) {};
  \node[label=right:{Im$\,\lambda$}] at (0,2) {};
  \node at (-2,0) {\rotatebox{90}{$-$}};
  \draw [red,very thick] (-5,0) to (0,0);
  \draw [red,very thick,domain=0:360] plot ({-2+cos(\x)},{sin(\x)});
  \node [label=below:{\textcolor{red}{$-1/\alpha$}}] at (-2,0) {};
  \draw [black!50, thick,->] (-2,0) -- (-2-0.7071,.7071) node [pos=0.9,above left]{\textcolor{red}{$\frac{\sqrt{1-\alpha \beta}}{\alpha}$}};
  \draw [red,thick,fill=white] (0,0) circle [radius=0.1];
 \end{tikzpicture}
 \end{center}

 {\bf Case $\alpha, \beta \neq 0$ and $\alpha \beta \geq 1$:} the negative real axis $ \Im \lambda =0, \Re \lambda < 0$. 
 \begin{center}
 \begin{tikzpicture}[scale=0.7]
  \draw[->,thick] (-5,0) to (1,0);
  \draw[->,thick] (0,-2) to (0,2);
  \node[label=below:{Re$\,\lambda$}] at (1,0) {};
  \node[label=right:{Im$\,\lambda$}] at (0,2) {};
  \draw [red,very thick] (-5,0) to (0,0);
  \draw [red,thick,fill=white] (0,0) circle [radius=0.1];
 \end{tikzpicture}
 \end{center}


\section{Discussion and future work} \label{sec:conc} 

We have established the characterization and synthesis of the response for
mass-spring networks with Rayleigh damping. In particular, our result shows that,
for each pair of $\alpha$ and $\beta$, there is a class of resonances that can
be realized when the damping matrix is $\bC = \alpha \bK + \beta \bM$. Clearly,
from section~\ref{sec:resonances}, when choosing different values of $\alpha$ and
$\beta$, it is possible to have any resonance with negative real part. Hence, if
we superimpose Rayleigh damped networks with different values of $\alpha$ and
$\beta$, it is possible to construct a network with any finite number of
resonances in the left half complex plane (provided they are real or come in
complex conjugate pairs). However, it is not clear whether the response of a
general damped network can be realized by superposing several Rayleigh networks
with different $\alpha$ and $\beta$. This is a question that we
plan to explore.

Another extension of this work would be to consider general damping by using
the quadratic eigenvalue problem \cite{Tisseur:2001:QEP,Gohberg:1982:MP}.
Specifically, we believe it is possible to use the spectral decomposition for
real symmetric quadratic pencils by \citet{Chu:2009:SDR} to characterize the
response of general damped networks. Then for the synthesis, we would need to
construct networks that realize the form of the response, but this is left for
future work.

\appendix

\section{Schur complement properties}
\label{app:schur}
Consider the partition of a matrix $\bA \in \complex^{n \times n}$ induced by the partition of $\{1,\ldots,n\}$ into two sets $I$ and $B$:
\[
 \bA = \begin{bmatrix} \bA_{BB} & \bA_{BI}\\ \bA_{IB} & \bA_{II} \end{bmatrix}.
\]
The Schur complement of the $II$ block in $\bA$ is defined as
\[
  \bS = \bA_{BB} - \bA_{BI} \bA_{II}^{-1} \bA_{IB},
\]
provided $\bA_{II}$ is invertible. The Schur complement is homogeneous of degree 1, since for nonzero $\lambda \in \complex$,
\[
 \lambda \bS = (\lambda \bA_{BB}) - (\lambda \bA_{BI}) (\lambda \bA_{II})^{-1} (\lambda \bA_{IB}).
\]

A quadratic form of the Schur complement is
equivalent to a quadratic form of the original matrix $\bA$, indeed by simple
manipulations we have 
\begin{equation}
 \bv_B^* \bS \bv_B = 
 \begin{bmatrix} \bv_B\\\bv_I \end{bmatrix}^*
 \bA
 \begin{bmatrix} \bv_B\\\bv_I \end{bmatrix},
 \text{~where~}
 \bv_I = - \bA_{II}^{-1} \bA_{IB} \bv_B.
 \label{eq:quadeq}
\end{equation}
A consequence of \eqref{eq:quadeq} is that if $\bA \in \real^{n\times n}$, $\bA \geq 0$ implies $\bS \geq 0$ (and similarly for the reverse inequality). Here the inequality $\bA \geq 0$ is understood as $\bA$ positive semi-definite.

The quadratic form $\bv^* \bA \bv$ for $\bA\in\complex^{n\times n}$ with $\bA^T = \bA$ (i.e., complex symmetric), can be written as
\[
 \begin{aligned}
 \Re (\bv^* \bA \bv) &= (\Re \bv)^T (\Re \bA) (\Re \bv) + (\Im \bv)^T (\Re \bA) (\Im \bv),\\
 \Im (\bv^* \bA \bv) &= (\Re \bv)^T (\Im \bA) (\Re \bv) + (\Im \bv)^T (\Im \bA) (\Im \bv).
 \end{aligned}
\]
By combining this fact with \eqref{eq:quadeq} we have that for complex symmetric $\bA$ and its Schur complement $\bS$:
\begin{equation} \label{eq:schursign}
 \Re \bA \geq 0 \Rightarrow \Re \bS \geq 0 \text{~and~} 
 \Im \bA \geq 0 \Rightarrow \Im \bS \geq 0,
\end{equation}
and similarly for the reverse inequalities.

\section*{Acknowledgments}
The authors are thankful to Graeme W.\ Milton, Daniel Onofrei and Pierre
Seppecher for insightful conversations on this subject. The work of AG was
supported through the University of Utah Mathematics Department VIGRE-REU
program. The work of FGV was partially supported by the National Science
Foundation DMS-0934664.

\bibliographystyle{abbrvnat}
\bibliography{propbib}

\begin{thebibliography}{14}
\providecommand{\natexlab}[1]{#1}
\providecommand{\url}[1]{\texttt{#1}}
\expandafter\ifx\csname urlstyle\endcsname\relax
  \providecommand{\doi}[1]{doi: #1}\else
  \providecommand{\doi}{doi: \begingroup \urlstyle{rm}\Url}\fi

\bibitem[Bott and Duffin(1949)]{Bott:1949:ISU}
R.~Bott and R.~J. Duffin.
\newblock Impedance synthesis without use of transformers.
\newblock \emph{Journal of Applied Physics}, 20:\penalty0 804, 1949.
\newblock \doi{10.1063/1.1698532}.

\bibitem[Camar-Eddine and Seppecher(2002)]{Camar:2002:CSD}
M.~Camar-Eddine and P.~Seppecher.
\newblock Closure of the set of diffusion functionals with respect to the
  {M}osco-convergence.
\newblock \emph{Math. Models Methods Appl. Sci.}, 12\penalty0 (8):\penalty0
  1153--1176, 2002.
\newblock ISSN 0218-2025.
\newblock \doi{10.1142/S0218202502002069}.

\bibitem[Camar-Eddine and Seppecher(2003)]{Camar:2003:DCS}
M.~Camar-Eddine and P.~Seppecher.
\newblock Determination of the closure of the set of elasticity functionals.
\newblock \emph{Arch. Ration. Mech. Anal.}, 170\penalty0 (3):\penalty0
  211--245, 2003.
\newblock ISSN 0003-9527.
\newblock \doi{10.1007/s00205-003-0272-7}.

\bibitem[Chu and Xu(2009)]{Chu:2009:SDR}
M.~T. Chu and S.-F. Xu.
\newblock Spectral decomposition of real symmetric quadratic
  {$\lambda$}-matrices and its applications.
\newblock \emph{Math. Comp.}, 78\penalty0 (265):\penalty0 293--313, 2009.
\newblock ISSN 0025-5718.
\newblock \doi{10.1090/S0025-5718-08-02128-5}.

\bibitem[Chung(2007)]{Chung:2007:GCM}
T.~J. Chung.
\newblock \emph{General continuum mechanics}.
\newblock Cambridge University Press, Cambridge, 2007.
\newblock ISBN 978-0-521-87406-9; 0-521-87406-8.

\bibitem[Curtis et~al.(1998)Curtis, Ingerman, and Morrow]{Curtis:1998:CPG}
E.~B. Curtis, D.~Ingerman, and J.~A. Morrow.
\newblock Circular planar graphs and resistor networks.
\newblock \emph{Linear Algebra Appl.}, 283\penalty0 (1-3):\penalty0 115--150,
  1998.
\newblock ISSN 0024-3795.
\newblock \doi{10.1016/S0024-3795(98)10087-3}.

\bibitem[Foster(1924{\natexlab{a}})]{Foster:1924:RT}
R.~M. Foster.
\newblock A reactance theorem.
\newblock \emph{The Bell System Technical Journal}, 3:\penalty0 259--267,
  1924{\natexlab{a}}.
\newblock ISSN 0005-8580.

\bibitem[Foster(1924{\natexlab{b}})]{Foster:1924:TRD}
R.~M. Foster.
\newblock Theorems regarding the driving-point impedance of two-mesh circuits.
\newblock \emph{The Bell System Technical Journal}, 3:\penalty0 651--685,
  1924{\natexlab{b}}.
\newblock ISSN 0005-8580.

\bibitem[Gohberg et~al.(1982)Gohberg, Lancaster, and Rodman]{Gohberg:1982:MP}
I.~Gohberg, P.~Lancaster, and L.~Rodman.
\newblock \emph{Matrix polynomials}.
\newblock Academic Press Inc. [Harcourt Brace Jovanovich Publishers], New York,
  1982.
\newblock ISBN 0-12-287160-X.
\newblock Computer Science and Applied Mathematics.

\bibitem[Guevara~Vasquez et~al.(2011)Guevara~Vasquez, Milton, and
  Onofrei]{Guevara:2011:CCS}
F.~Guevara~Vasquez, G.~W. Milton, and D.~Onofrei.
\newblock Complete characterization and synthesis of the response function of
  elastodynamic networks.
\newblock \emph{J. Elasticity}, 102\penalty0 (1):\penalty0 31--54, 2011.
\newblock ISSN 0374-3535.
\newblock \doi{10.1007/s10659-010-9260-y}.

\bibitem[Milton and Seppecher(2008)]{Milton:2008:RRM}
G.~W. Milton and P.~Seppecher.
\newblock Realizable response matrices of multi-terminal electrical, acoustic
  and elastodynamic networks at a given frequency.
\newblock \emph{Proc. R. Soc. Lond. Ser. A Math. Phys. Eng. Sci.}, 464\penalty0
  (2092):\penalty0 967--986, 2008.
\newblock ISSN 1364-5021.
\newblock \doi{10.1098/rspa.2007.0345}.

\bibitem[Milton and Seppecher(2010{\natexlab{a}})]{Milton:2010:EC}
G.~W. Milton and P.~Seppecher.
\newblock Electromagnetic circuits.
\newblock \emph{Netw. Heterog. Media}, 5\penalty0 (2):\penalty0 335--360,
  2010{\natexlab{a}}.
\newblock ISSN 1556-1801.
\newblock \doi{10.3934/nhm.2010.5.335}.

\bibitem[Milton and Seppecher(2010{\natexlab{b}})]{Milton:2010:HEC}
G.~W. Milton and P.~Seppecher.
\newblock Hybrid electromagnetic circuits.
\newblock \emph{Physica B: Condensed Matter}, 405\penalty0 (14):\penalty0 2935
  -- 2937, 2010{\natexlab{b}}.
\newblock ISSN 0921-4526.
\newblock \doi{10.1016/j.physb.2010.01.007}.
\newblock Proceedings of the Eighth International Conference on Electrical
  Transport and Optical Properties of Inhomogeneous Media, ETOPIM-8.

\bibitem[Tisseur and Meerbergen(2001)]{Tisseur:2001:QEP}
F.~Tisseur and K.~Meerbergen.
\newblock The quadratic eigenvalue problem.
\newblock \emph{SIAM Rev.}, 43\penalty0 (2):\penalty0 235--286, 2001.
\newblock ISSN 0036-1445.
\newblock \doi{10.1137/S0036144500381988}.

\end{thebibliography}

\end{document}